\documentclass[preprint,12pt]{elsarticle}

\usepackage[T1]{fontenc}

\usepackage{algpseudocode}
\usepackage{algorithm}
\usepackage{amsmath}
\usepackage{amssymb}
\usepackage{amsthm}
\usepackage{amsbsy}

\usepackage{comment}
\usepackage{hyperref}

\newtheorem{definition}{Definition}[section]
\newtheorem{theorem}{Theorem}[section]
\newtheorem{lemma}{Lemma}[section]
\newtheorem{corollary}{Corollary}[section]
\newtheorem{proposition}{Proposition}[section]

\newtheorem{claim}{Claim}[section]

\usepackage{mathtools}
\usepackage[algo2e,ruled,vlined]{algorithm2e}
\usepackage{algcompatible}

\usepackage{todonotes}
\usepackage[a4paper, lmargin=0.1666\paperwidth, rmargin=0.1666\paperwidth, tmargin=0.1111\paperheight, bmargin=0.1111\paperheight]{geometry} 
\usepackage{parskip}

  \baselineskip 8 mm

\journal{Discrete Optimization}

\begin{document}

\begin{frontmatter}
\title{Disjoint Dominating    and $2$-Dominating Sets in  Graphs:  Hardness and Approximation results}
\author[a1]{Soumyashree Rana}
\ead{maz218122@iitd.ac.in}
\address[a1]{Department of Mathematics, Indian Institute of Technology Delhi, New Delhi, 110016, India}
\author[a2]{Sounaka Mishra}
\ead{sounak@iitm.ac.in}
\address[a2]{Department of Mathematics, Indian Institute of Technology Madras, Chennai, 600036, India} 

\author[a1]{Bhawani Sankar Panda}
\ead{bspanda@maths.iitd.ac.in}

\begin{abstract}
A set $D \subseteq V$ of a graph $G=(V, E)$ is a dominating set of $G$ if each vertex $v\in V\setminus D$ is adjacent to at least one vertex in $D,$ whereas a set $D_2\subseteq V$ is a $2$-dominating (double dominating) set of $G$ 
if each vertex $v\in V\setminus D_2$ is adjacent to at least two vertices in $D_2.$ A graph $G$ is a $DD_2$-graph if there exists a pair ($D, D_2$) of dominating set and $2$-dominating set of $G$ which are disjoint.   
In this paper, we  solve  some  open problems posed  by M.Miotk, J.~Topp and P.{\.Z}yli{\'n}ski (Disjoint dominating and 2-dominating sets in graphs,  Discrete Optimization, 35:100553, 2020) by giving  approximation algorithms  for the problem of determining a minimal spanning $DD_2$-graph of  minimum size (\textsc{Min-$DD_2$}) with an approximation ratio of $3$; a minimal spanning $DD_2$-graph of  maximum size (\textsc{Max-$DD_2$}) with an approximation ratio of $3$; and for the problem of  adding  minimum  number of edges  to a graph $G$  to make it a $DD_2$-graph  (\textsc{Min-to-$DD_2$}) with an $O(\log n)$ approximation ratio. Furthermore, we prove that \textsc{Min-$DD_2$} and \textsc{Max-$DD_2$} are \textsc{APX}-complete for graphs with maximum degree $4.$
We also  show that \textsc{Min-$DD_2$} and \textsc{Max-$DD_2$} are approximable within a factor of $1.8$ and $1.5$ respectively, for any $3$-regular graph. Finally, we show the inapproximability result of \textsc{Max-Min-to-$DD_2$} for bipartite graphs, that this problem can not be approximated within $n^{\frac{1}{6}-\varepsilon}$ for any $\varepsilon >0,$ unless \textsc{P=NP}.   
\end{abstract}
\end{frontmatter}

\section{Introduction}
Let $G=(V, E)$ be a finite, simple, and undirected graph with vertex set $V$ and edge set $E.$ The graphs considered in this paper are  without isolated vertices.  A set $D \subseteq V$ is said to be a \textit{dominating set} of $G$ if each vertex in $V\setminus D$ has an adjacent vertex in $D.$ The minimum cardinality among all dominating sets of $G$ is the \textit{domination number} of $G,$ and it is denoted by $\gamma(G).$ Likewise, a set $D_2 \subseteq V$ is known as a \textit{$2$-dominating (double dominating) set} of $G$ if each vertex in $V\setminus D_2$ has at least two adjacent vertices in $D_2.$ The minimum cardinality among all $2$-dominating sets of $G$ is the \textit{$2$-domination number} of $G,$ and it is denoted by $\gamma_2(G).$ 

A graph having no isolated vertex contains two disjoint dominating sets, which was first observed by Ore \cite{ore1962theory}. This implies that the vertex set of a graph can be partitioned into two disjoint dominating sets, provided it has no isolated vertices.
In the near past, researchers have studied computing the minimum size of a pair of disjoint dominating sets in a graph. It is known to be \textsc{NP}-complete \cite{henning2009remarks}. Some other results related to this graph parameter are available in  \cite{anusuya2012note,hedetniemi2006disjoint}. Conditions that guarantee the existence of a dominating set whose complement contains a $2$-dominating set, paired dominating set, or an independent dominating set are presented in \cite{henning2010independent,henning2013graphs,kiunisala2016pairs,southey2011characterization}. Henning and Rall \cite{henning2013graphs} initiated the study of graphs having a dominating set whose complement is a $2$-dominating set. 
In a graph $G$, a pair of disjoint sets $(X, Y)$ is called a $DD_2$-pair if $X$ is a dominating set and $Y$ is a 2-dominating set of $G$. A graph $G$ is called $DD_2$-graph if it has a $DD_2$-pair.
It is easy to verify that complete graphs $K_n$ with $n \geq 3$, cycles $C_n$ with $n \geq 3$, and paths $P_n$ with $n=3$ or $n \geq 5$ are $DD_2$-graphs. It is known that one can construct a non-$DD_2$ graph by adding a pendant edge to every vertex of a given graph $G$. It is important to mention here that if we add at least two pendant edges to each vertex of a graph $G$, then the resulting graph is always a $DD_2$-graph. It is also proved that any graph $G$ with a minimum degree of at least 2 is a $DD_2$-graph \cite{henning2010independent}.  

Miotk et al. \cite{miotk2020disjoint} further continued the study of conditions that ensure a partition of the vertex set of a graph into a dominating set and a $2$-dominating set from algorithmic insights. They have considered various optimization problems associated with $DD_2$-graphs. 
A graph $H=(V_H, E_H)$ is called a spanning subgraph of $G=(V, E)$ if $V_H=V$ and $E_H \subseteq E$. A connected graph $G$ is said to be a \textit{minimal $DD_2$-graph} if $G$ itself is a $DD_2$-graph and no proper spanning subgraph of $G$ is a $DD_2$-graph. We say a disconnected graph $G$ is a \textit{minimal $DD_2$-graph} if every connected component of $G$ is a minimal $DD_2$-graph.    
In \cite{miotk2020disjoint}, the authors considered the computational complexity of some of the following optimization problems related to $DD_2$-graph property. These problems, except \textsc{Max-Min-to-$DD_2$}, were  defined in \cite{miotk2020disjoint}. In this paper, we initiate the study of the \textsc{Max-Min-to-$DD_2$} problem. 
\begin{description}
  \item[1.] \textsc{Min-$DD_2$} : Given a $DD_2$-graph $G=(V, E),$ in \textsc{Min-$DD_2$}, it is required to find a subgraph $H=(V, E')$ of $G$ such that $H$ is a minimal spanning $DD_2$-graph with minimum $|E'|$.
  
  \item[2.] \textsc{Max-$DD_2$} : Given a $DD_2$-graph $G=(V, E),$ in \textsc{Max-$DD_2$}, it is required to find a subgraph $H=(V, E')$ of $G$ such that $H$ is a minimal spanning $DD_2$-graph with maximum $|E'|.$
  
  \item[3.] \textsc{Min-to-$DD_2$} : Given a non-$DD_2$-graph $G=(V, E),$ in \textsc{Min-to-$DD_2$}, it is required to find a minimum size edge set $E'$ such that $E \cap E' = \emptyset$ and $(V, E \cup E')$ is a minimal spanning $DD_2$-graph.
  
  \item[4.] \textsc{Max-Min-to-$DD_2$} : Given a non-$DD_2$-graph $G=(V, E),$ in \textsc{Max-Min-to-$DD_2$}, it is required to find a maximum size edge set $E'$ such that $E \cap E' = \emptyset$ and $(V, E \cup E')$ is a minimal spanning $DD_2$-graph.
\end{description}

In \cite{miotk2020disjoint}, it is proved that both \textsc{Min-$DD_2$} and \textsc{Max-$DD_2$} are \textsc{NP}-complete. They have also proved that \textsc{Min-to-$DD_2$} is \textsc{NP}-complete. In this paper, we extend the algorithmic study of these optimization problems. The main contributions of the paper are summarised below.
\begin{itemize}
\item[1.] We prove that \textsc{Min-$DD_2$} and \textsc{Max-$DD_2$} are approximable within a factor of $3$.

\item[2.] For any $3$-regular graph, we show that the \textsc{Min-$DD_2$} and \textsc{Max-$DD_2$} are approximable within a factor of $1.8$ and $1.5$ respectively.

\item[3.] Moreover, we prove that the \textsc{Min-$DD_2$} and the \textsc{Max-$DD_2$} are \textsc{APX}-complete for graphs with maximum degree $4.$ 

\item [4.]We provide an $O(\log n)$ factor approximation algorithm for \textsc{Min-to-$DD_2$}.

\item[5.]  We show that the \textsc{Max-Min-to-$DD_2$} can not be approximated within  a factor of $n^{\frac{1}{6}-\varepsilon}$ for any $\varepsilon >0,$ unless \textsc{P=NP}, for bipartite graphs.
\end{itemize}

The above results $(1)$ and $(4)$ answers the open problems posed in \cite{miotk2020disjoint}.

\section{Preliminaries}
This section gives some pertinent definitions and states some preliminary results that will be used in this paper. 

Let $G=(V, E)$ be a finite, simple, and undirected graph with no isolated vertex. The open neighborhood of a vertex $v$ in $G$ is $N_G(v)=\{u \in V \mid uv \in E\}$ and the closed neighborhood is $N_G[v]= N_G(v) \cup \{v\}.$ The degree of a vertex $v$ in $G$ is $|N_G(v)|$ and is denoted by $d_G(v).$ If $d_G(v)=1$, then $v$ is called a \textit{pendent vertex} (leaf), and its unique neighbor $u$ in $G$ is called the \textit{support vertex}. We will also denote $L_G(v)$ as the set of pendant neighbors of $v$ in $G$. If a support vertex has at least two pendant vertices, then it is a strong support vertex, otherwise, it is a weak support vertex (only one pendant neighbor). The set of weak and strong support vertices of a graph $G$ are denoted by $S'_G$ and $S''_G,$ respectively. We will denote $L'_G$ and $L''_G$ as the set of pendant neighbors of weak support vertices and strong support vertices in $G$, respectively. Note that $S_G = S'_G \cup S''_G$ and $L_G = L'_G \cup L''_G$.
The minimum and maximum degree of $G$ is denoted by $\delta(G)$ and $\Delta(G),$ respectively. For $D\subseteq V,~G[D]$ denotes the subgraph induced by $D$ on $G.$ We use the notation $[k]$ for $\{1, 2, \cdots, k\}.$  We refer to \cite{west2001introduction} for other notations and graph terminologies that are not mentioned here.

A \textit{bipartite graph} is a graph $G=(V,E)$ whose vertices can be partitioned into two disjoint sets $X$ and $Y$ such that every edge has one endpoint in $X$ and other in $Y.$ Sometime, we denote a bipartite graph with bi-partition $X$ and $Y$ of $V$ as $G=(X\cup Y,E).$ 
An edge set $M\subseteq E$ in $G$ is called a $\it{matching}$ if the degree of each vertex in $(V, M)$ is at most 1, and the vertices $V_M\subseteq V$ of degree 1 in $(V, M)$ are called \textit{$M$-saturated} vertices in $G$. If $G[V_M] = (V_M, M)$ then $M$ is called an \textit{induced matching} in $G$. A vertex set $S$ is called a \textit{vertex cover} in $G$ if every edge $uv \in E$ has at least one end-vertex in $S$. 
A multigraph $H$ is called a \textit{corona graph} if every vertex of $H$ is either a leaf or it is adjacent to a leaf of $H.$ The \textit{subdivision graph S(H)} of a multigraph $H$ is the graph obtained from $H$ by inserting a new vertex onto each edge of $H.$

The following known results will be used throughout the paper.  

\begin{lemma}\cite{miotk2020disjoint}
\label{minimal-DD-2size}
Any minimal $DD_2$-graph has at least $3$ vertices and the $3$-vertex path $(P_3)$ is the smallest minimal $DD_2$-graph. Thus, a minimal spanning $DD_2$-graph of a graph $G$ of order $n$ must be of size at least $\frac{2n}{3}.$  
\end{lemma}

\begin{proposition}\cite{miotk2020disjoint}
\label{DD-2-bipartite}
A connected graph $G$ is a $DD_2$-graph if and only if $G$ has a spanning bipartite subgraph $T=(A,B,E_T)$ such that $d_T(a)\geq 2$ for every $a\in A,$ and $d_T(b)\geq 1$ for every $b\in B.$
\end{proposition}

\begin{lemma}\cite{henning2013graphs} \label{lemma2.2}
    Every graph $G$ with $\delta(G)\geq 2,$ is a $DD_2$-graph.
\end{lemma}

\begin{proposition}\cite{miotk2020disjoint} \label{characterization-DD-2}
A connected graph $G$ is a minimal $DD_2$-graph if and only if $G$ is a star $(K_{1,n}, n\geq 2)$, a cycle $C_4,$ or a subdivision graph $S(T)$ of a connected corona multigraph $T$. 
\end{proposition}

\begin{lemma}\cite{miotk2020disjoint} \label{characterization-subdivision}
A connected graph $G$ is the subdivision graph of a corona graph if and only if $G$ is a bipartite graph $(A \cup B, E_G)$ such that $d_G(a)=2$ for every vertex $a\in A$ and every vertex $b\in B$ is a leaf or it is at a distance two from some leaf of $G.$ 
\end{lemma}


\begin{proposition}\cite{miotk2020disjoint} 
\label{charac-DD-2-weaksupport-certified}
 Let $G$ be a graph with no isolated vertex. Then the following statements are equivalent:
 \begin{enumerate}
   \item $G$ is a $DD_2$-graph.
   \item $N_G(s)\setminus (L_G\cup S_G)\neq \emptyset$ for every weak support vertex $s$ of $G.$
 \end{enumerate}
\end{proposition}

Using the above result, whether a given graph $G=(V,E)$ is a $DD_2$-graph can be recognized in polynomial time. 

\begin{lemma}\cite{zito1999induced} \label{max-matching}
    If $G$ is a $(\delta,\Delta)$-graph with $n$ vertices then the maximum cardinality induced matching in $G$, $|M|\leq \frac{n\Delta}{2(\Delta+\delta-1)}.$ 
\end{lemma}

\section{Approximation Algorithms}
In this section, we design polynomial time approximation algorithms for \textsc{Min-$DD_2$} and \textsc{Max-$DD_2$}. 

\begin{theorem}
\textsc{Min-$DD_2$} is approximable within a factor of $3$.
\end{theorem}
\begin{proof}
Let $G=(V, E)$ be a given instance of \textsc{Min-$DD_2$}. Let ${E}^* \subseteq E$ such that $(V, {E}^*)$ is a minimal spanning $DD_2$-subgraph of $G$ and $|{E}^*|$ is of minimum cardinality. By Lemma \ref{minimal-DD-2size}, we have $|{E}^*| \geq \frac{2n}{3}$.

Since $G$ is a $DD_2$-graph, by Proposition \ref{DD-2-bipartite} there exists a bipartition $V=A \cup B$ such that $H=(A \cup B,E_{H})$ is a spanning bipartite subgraph of $G$, $d_{H}(a)\geq 2$ for every $a\in A,$ and $d_{H}(b)\geq 1$ for every $b\in B.$ Such a bipartition of $V$ can be computed in polynomial time \cite{miotk2020disjoint}.
It can be easily seen that $A$ and $B$ are disjoint dominating set and a $2$-dominating set of $G,$ respectively. Thus, $|B|\geq 2$ implies $|A|\leq n-2.$

From Proposition \ref{characterization-DD-2}, the minimal spanning $DD_2$-subgraph $H$ consists of components of the form $K_{1,n_1}, 2 \leq n_1 \leq n$, a cycle $C_4$ or a subdivision graph of a corona multigraph, $S(T)$. Note that the number of edges in a component isomorphic to $K_{1,n_1}$ is $n_1 \leq 2((n_1+1)-2), \forall n_1\geq 2,$ and the number of edges in a component isomorphic to $C_4$ is $4\leq 2(|V(C_4)|-2).$ Also, any component that is a subdivision of a corona multigraph $S(T)$ has exactly $2$ edges to $B,$ due to Lemma \ref{characterization-subdivision}. Thus, for any $S(T)$ component, $2|A|\leq 2(|V(S(T))|-2)$. Since the sum of the number of edges in each component of $H$ with $n'$ vertices is at most $2(n'-2),$ we have  $|E_H|\leq 2(n-2).$ 

Now, $\frac{|E_H|}{|{E}^*|}\leq 2(n-2)\times \frac{3}{2n}<3.$ Therefore, the \textsc{(Min-$DD_2$)} is approximable within a factor of $3$ for any graph.
\end{proof}

\begin{theorem}
\textsc{Max-$DD_2$} is approximable within a factor of $3$.
\end{theorem}
\begin{proof}
Given a $DD_2$-graph $G=(V,E)$ (an instance of \textsc{Max-$DD_2$}), let $H^*=(V, E^*)$ be a maximum size minimal spanning $DD_2$-subgraph of $G.$ By Proposition \ref{characterization-DD-2}, each component of $H^*$ is isomorphic to $K_{1,n_1}, n_1\geq 2,$ a cycle $C_4$ or a subdivision graph of a connected corona multigraph.
    
Now, we claim that $|E^*| < 2n$, where $n$ is the number of vertices in the graph $G.$ It is easy to observe that the number of edges in $K_{1,n_1}$ component is $n_1\leq (n_1+1), \forall n_1\geq 2,$ and in $C_4$ component the number of edges is $4\leq |V(C_4)|.$ It remains to show that the subdivision graph $G'=(V', E')$ of a connected corona multigraph having $n'$ vertices have at most $2n'$ number of edges. We prove this in the following constructive way. First, delete a minimum number of vertices (say $k$) such that the graph $G'$ becomes a tree (say $T$). 
Note that the deleted vertices must be the vertices that correspond to the subdivided edge of the multigraph because deleting the original vertex disconnects the graph.
Thus, $|E(T)|=|V(T)-1|=n'-k-1.$ Since these $k$ vertices contribute exactly $2$-edges to the minimal $DD_2$-subgraph $G',$ $|E(G')|=|E(T)|+2k=n'-k-1+2k=n'+k-1< 2n',$ as $k<n'.$ Moreover, the sum of the number of edges in each component is at most $2n,$ which implies $|E^*| < 2n.$

Let $H=(V, E_H)$ be a minimal spanning $DD_2$-subgraph of $G$. Then
by Lemma \ref{minimal-DD-2size}, $|E_H|\geq \frac{2n}{3}$. Thus, $\frac{|E^*|}{|E_H|} < 2n\times \frac{3}{2n}=3.$  
    
Therefore, \textsc{Max-$DD_2$} is approximable within a factor of $3$ for any graph.
\end{proof}

We conclude the above two theorems with this corollary stated below.
\begin{corollary}\label{3-reg-APX}
 \textsc{Min-$DD_2$} and \textsc{Max-$DD_2$} are in class \textsc{APX} for any graph $G.$    
\end{corollary}

Next, we improve these approximation factors for \textsc{Min-$DD_2$} and \textsc{Max-$DD_2$}, for 3-regular graphs.

\begin{algorithm2e}
\textbf{Input:} A $3$-regular $DD_2$-graph $G=(V,E).$\\
\textbf{Output:} A minimum size minimal spanning $DD_2$-subgraph $H$ of $G.$\\
\Begin{
Let $V' = V$ and $E'=E$;\\
$A'=\emptyset$;\\
\While{$\exists$ an edge $uv\in E'$}{
    $A'=A'\cup \{u,v\};$\\
    Delete all the vertices in $N[u]\cup N[v]$ and all the edges incident on them from $(V', E')$;
  }
  Let $T$ be the remaining vertices in $V'$;\\
  $A=A'\cup T$;\\
  $B=V\setminus A;$\\
  Construct edge set $X$ by selecting one edge incident on each vertex in $T$;\\
  $E_H=\{uv \in E \mid u\in A, v\in B\} \setminus X;$\\
  \Return $H=(A\cup B, E_H)$;
}
\caption{\textsc{Approx-Min-$DD_2$}}
\label{Min-DD-2}
\end{algorithm2e}

\begin{theorem}
For 3-regular graphs, \textsc{Min-$DD_2$} is approximable within a factor of $1.8$.
\end{theorem}
\begin{proof}
Given a 3-regular $DD_2$-graph $G=(V,E)$, let $E^*\subseteq E$ such that $(V, E^*)$ is a minimal spanning $DD_2$-subgraph of $G$ and $|E^*|$ is minimum among all such subgraphs of $G$.
By Lemma \ref{minimal-DD-2size}, we have $|E^*| \geq \frac{2n}{3}.$ 

Let $H=(A\cup B, E_H)$ be the graph returned by Algorithm \ref{Min-DD-2}. Now, we claim that $H$ is a minimal spanning $DD_2$-subgraph of $G$. From the construction of set $A'$ in Algorithm \ref{Min-DD-2}, it follows that the edge set in $G[A']$ forms a maximal induced matching $M$ in $G$. This implies that each vertex in $A'$ has exactly two neighbors in $B$. Since $T= V \setminus (A' \cup N(A'))$, it follows that $T$ is an independent set in $G$ and all its neighbors are in $B$. It is important to note that $E_H$ does not contain any edge from the edge set $X$. This implies that each vertex in $A$ has exactly two neighbors in $B$. Now it is easy to observe that each vertex in $B$ has at least one neighbor in $A$.    
Hence $H$ is a spanning $DD_2$-subgraph of $G$ ($A$ is a dominating set and $B$ is a 2-dominating set of $G$). Since each vertex in $A$ has exactly two neighbors in $B$, removing any edge from $E_H$ would result in a non-$DD_2$-subgraph $H$ of $G,$ as $B$ would not be a 2-dominating set of $G.$ Hence the claim.

Now, let $a=|A'|$ and $t=|T|.$ As $G$ is a $3$-regular graph, $|E|=\frac{3n}{2}.$ Since every vertex of $A'$ has exactly 2 neighbors in $B$ and every vertex of $T$ has exactly 3 neighbors in $B,$ before deleting one edge incident on each vertex of $T,$ the total number of edges incident on $A$ is 
\begin{equation}
\label{equ1}
  \frac{a}{2}+2a+3t\leq \frac{3n}{2} \Rightarrow t \leq \frac{1}{2}(n-\frac{5}{3}a).  
\end{equation}

By Lemma \ref{max-matching}, $|M|\leq \frac{3n}{10}$. Moreover, every edge of the maximal induced matching $M,$ contributes exactly $2$ vertices to $A'.$ Thus, 
\begin{equation}
\label{equ2}
    \frac{a}{2} \leq \frac{3n}{10} \Rightarrow a\leq \frac{3n}{5}.
\end{equation}

Since, in subgraph $H,$ each vertex in $A$ has exactly two neighbors in $B,$ we have $|E_H|=2a+2t.$ Using Equation (\ref{equ1}) and (\ref{equ2}), we obtain,
\begin{equation*}
   |E_H|\leq 2a+n-\frac{5}{3}a=n+\frac{1}{3}a \leq n+\frac{n}{5} =\frac{6n}{5}.
\end{equation*}
Thus, $\frac{|E_H|}{|E^*|}\leq \frac{6n}{5} \times \frac{3}{2n} =1.8.$ Therefore, \textsc{Min-$DD_2$} is approximable within a factor of $1.8$ for $3$-regular graphs.
\end{proof}

\begin{theorem}
\label{MaxDD_2APX}
For $3$-regular graphs, \textsc{Max-$DD_2$} is approximable within a factor of $1.5$. 
\end{theorem}
\begin{proof}
Given a 3-regular $DD_2$-graph $G=(V,E)$, let $H^*=(V, E^*)$ be a maximum size minimal spanning $DD_2$-subgraph of $G.$ From Proposition \ref{characterization-DD-2} it follows that the components of $H^*$ consist of minimal $DD_2$-graphs isomorphic to $K_{1,n_1}, 2 \leq n_1 \leq n,$ a cycle $C_4$ or a subdivision graph of a connected corona multigraph.
    
We claim that $|E^*|\leq n$, where $n$ is the number of vertices in the graph $G.$ It is easy to observe that the number of edges in $K_{1,n_1}$ component is $n_1\leq |V(K_{1,n_1})|, \forall n_1\geq 2,$ and the number of edges in $C_4$ component is $4\leq |V(C_4)|.$

Now it is left to show that the subdivision graph $G'=(V', E')$ of a connected corona multigraph has at most $|V'|$ number of edges. Since $G$ is a $3$-regular graph, it is easy to see that no two cycles intersect at any vertex $v$, otherwise $d(v)\geq 4.$ We claim that no connected component contains two distinct induced cycles. 
This claim will imply that $G'$ has at most one cycle and hence $|E'| \leq |V'|$. Now, we prove this claim.
On the contrary, let $C^1$ and $C^2$ be the two induced cycles in one connected component (say $C$) and let $P$ be the path connecting a vertex of $C^1$ to a vertex of $C^2.$ By Lemma \ref{characterization-subdivision}, $C$ is a bipartite graph $C = (A_C \cup B_C, E_C)$ such that $d_C(a)=2$ for every vertex $a\in A_C$ and every vertex $b\in B_C$ is a leaf or it is at a distance two from some leaf of $C.$ Thus the two endpoints of the path $P$ are in $B_C,$ (say $b_i,b_j\in B_C$). Since every path in $C$ alternate between a vertex from $A_C$ and $B_C,$ the vertices next to $b_i$ (or $b_j$) in path $P$ must be a vertex in $A_C$. Hence, the vertex $b_i$ (or $b_j$) is of degree at least 4, which is a contradiction. Thus, no connected component which is a subdivision graph of a corona multigraph contains two distinct induced cycles. 

Henceforth, summing the number of edges of all the components of $H^*$ is at most $n,$ which implies $|E^*| \leq n.$
     
Let $H=(V_H,E_H)$ be any minimal spanning $DD_2$-subgraph of $G.$ Thus, by Lemma \ref{minimal-DD-2size}, the size of any minimal spanning $DD_2$-graph is at least $\frac{2n}{3}$, so we have $|E_H|\geq \frac{2n}{3}.$ 
 
Hence, $$\frac{|E^*|}{|E_H|} \leq n\times \frac{3}{2n}=1.5.$$ Therefore, \textsc{Max-$DD_2$} is approximable within a factor of $1.5$ for $3$-regular graph.
\end{proof}

Next, we consider the inapproximability of \textsc{Min-to-$DD_2$}. However, it is known to be \textsc{NP}-complete \cite{miotk2020disjoint}.

\begin{proposition} \cite{miotk2020disjoint} \label{NP-C-gen}
 \textsc{Min-to-$DD_2$} is \textsc{NP}-complete for general graphs.   
\end{proposition}


The authors in \cite{miotk2020disjoint} proved the above proposition by establishing a cost-preserving reduction from \textsc{Set Cover}.  By using the lower bound result of \textsc{Set Cover} \cite{alon2006algorithmic}, we can have the following lower bound result for \textsc{Min-to-$DD_2$}. 

\begin{corollary}
For any graph $G=(V,E)$, \textsc{Min-to-$DD_2$} can not be approximated within $(1-\varepsilon)\log |V|$ for any $\varepsilon >0,$ unless \textsc{P=NP}.
\end{corollary}

In order to design an $O(\log n)$-factor approximation algorithm for 
\textsc{Min-to-$DD_2$}, we need similar approximation algorithms for 
\textsc{Min-W-Dom-Set} and \textsc{Min-W-T-Dom-Set}. Given a graph $G$ and a vertex weight function $w: V(G)\rightarrow \mathbb{N}$, in minimum weighted dominating set problem (\textsc{Min-W-Dom-Set}) it is required to find a dominating set $D\subseteq V$ such that $\sum_{v\in D} w(v)$ is minimum.  
Given a graph $G=(V, E)$ and a vertex set $T \subseteq V$, a vertex set $D$ is called a $T$-dominating set if $D \cap N[v] \neq \emptyset$, for each vertex $v \in T$.
Given a graph $G,$ a vertex weight function $w$, and a set $T\subseteq V$, in minimum weighted $T$ dominating set problem (\textsc{Min-W-T-Dom-Set}) it is required to find a $T$-dominating set $D\subseteq V$ of minimum weight.

\begin{theorem} \cite{vazirani2001approximation} \label{min_weight_dom_set_approx}
 \textsc{Min-W-Dom-Set} is approximable within a factor of $O(\log n).$
\end{theorem}

Next, we design an $O(\log n)$-factor approximation algorithm for \textsc{Min-W-T-Dom-Set} and it will be used to design a similar approximation algorithm for \textsc{Min-to-$DD_2$}.

\begin{theorem} \label{approx_T-dom}
 \textsc{Min-W-T-Dom-Set} can be approximated with an approximation ratio of $O(\log n).$
\end{theorem}
\begin{proof}
We prove it by establishing a reduction to \textsc{Min-W-Dom-Set}.  Given an instance $(G, w, T)$ of \textsc{Min-W-T-Dom-Set}, we construct an instance $(G', w')$ of \textsc{Min-W-Dom-Set} as follows. The graph $G'=(V', E')$ is obtained from $G$ by introducing three new vertices $t, p, q$ and the edges $\{(t, p), (t, q)\}\cup \{(t, v) \mid v \in V\setminus T\}$. Next, we define $w'(v)=w(v)$, for all $v \in V$, $w'(t)=1$, $w'(p)=w'(q)=|V|$.

Let $D$ be a $T$-dominating set in $G$. Then it is easy to observe that $D'=D \cup \{t\}$ is a dominating set in $G'$. Also, $w'(D') = w(D) +1$.

From the construction of $G'$, it is easy to observe that every minimum weight dominating set in $G'$ must contain the vertex $t$. Otherwise, the vertices $p$ and $q$ must be in the dominating set $D$, then the set $(D\cup\{t\}) \setminus \{p, q\}$ is a dominating set in $G'$ with strictly smaller weight than $D$. Because of this property of a minimum weight dominating set in $G'$, without loss of generality, we will consider the minimal dominating sets in $G'$ that contain the vertex $t$. Now, it is easy to observe that if $D'$ is a minimal dominating set in $G'$ then $D = D'\setminus \{t\}$ is a minimal $T$-dominating set in $G$. Thus $D$ is a dominating set in $G$ as $t$ is not adjacent to any vertex in $T$ and $(D'\setminus \{t\}) \cap T \neq \emptyset$. Also, $w(D) = w'(D')-1$.

From the above discussion, it follows that $w'(D'_{opt}) = w(D_{opt}) +1$, where $D'_{opt}$ is a minimum weight dominating set of $G'$ and $D_{opt}$ is a minimum weight $T$-dominating set of $G$. 

Now, for any minimal dominating set $D'$ of $G'$, we have $\frac{w(D)}{w(D_{opt})} \leq 2 \frac{w(D)+1}{w(D_{opt})+1} = 2 \frac{w'(D')}{w'(D'_{opt})} \leq O(\log |V'|)= O(\log |V|).$
\end{proof}

\begin{theorem}
 \textsc{Min-to-$DD_2$} can be approximated within an approximation ratio of $O(\log n)$. 
\end{theorem}
\begin{proof}
Given a non-$DD_2$-graph $G=(V, E),$ let $S_G=S'_G\cup S''_G$ be the set of all support vertices of $G$, where $S'_G$ and $S''_G$ be the set of weak support vertices and strong support vertices of $G,$ respectively. Let $L_G=L'_G\cup L''_G$ be the set of all pendant vertices of $G,$ where $L'_G$  and $L''_G$ be the set of pendants adjacent to weak support vertices and strong support vertices of $G,$ respectively. 
Based on these notations we design Algorithm \ref{Apx-Min_to_dd_2} for \textsc{Min-to-$DD_2$}.


\begin{algorithm2e}[]
\textbf{Input:} A non-$DD_2$-graph $G=(V,E).$\\
\textbf{Output:} A solution $\hat{E}$ to \textsc{Min-to-$DD_2$} for $G.$\\
\Begin{
 Compute the following sets\\
  $U=\{v\in S'_G \mid N(v)\setminus (L_G\cup S_G)\neq \emptyset\}$; \\
 $A=S'(G) \setminus U$; \\
 $B=N(A)\cap S''(G)$;\\
 Construct the graph $H=(A \cup B, E')$ with the edge set $E'$ as the set of edges in $G[A \cup B]$ minus the edges in $G[B]$;\\
 Define the vertex weight function on $H$ as $w(v)$ as the number of pendant neighbors of $v$ in $G$;\\
 Compute a weighted $T$-dominating set $S$ of $H$ where $T=A$;\\
 Let $K$ be the set of pendant neighbors of $S$ in $G$;\\  
 $\hat{E}$ be the minimum size edge set $uv$  such that $u,v\in K$ and degree of each vertex $v \in K$ is at least 2 in $(V, E \cup \hat{E})$;\\
\Return $\hat{E}$;} \caption{\textsc{Approx-Min-to-$DD_2$}}
 \label{Apx-Min_to_dd_2}
\end{algorithm2e}

The correctness of this algorithm is based on the characterization described in Proposition \ref{charac-DD-2-weaksupport-certified}. This characterization of $DD_2$-graphs suggests that if $N(v) \cap (L_G \cup S_G) = \emptyset$, for a vertex $v \in S'_G$, then we can make the set $N(v) \cap (L_G \cup S_G)$ non-empty by introducing new edges in the following ways: $(a)$ add an edge to the unique pendant neighbor of $v$, $(b)$ add a new edge such that $v$ has a new neighbor $z$ which is not in $L_G \cup S_G$, or $(c)$ add new edges to all the pendant neighbors of a vertex $z$ in $N(v) \cap S''_G$. 
The main idea of this algorithm is based on adding new edges so that $(a)$ or $(c)$ satisfies. It is easy to observe that if we are adding an edge to a pendant neighbor of $v \in S'_G$, then $N(u) \cap (L_G \cup S_G)$ becomes nonempty for all the neighbors $u \in N(v) \cap S'_G$. This observation suggests constructing the vertex-weighted graph $H$ and looking for a $T$-dominating set in it. Now, it is easy to observe that $\hat{E}$ is a feasible solution to \textsc{Min-to-$DD_2$} for $G$. In fact, $\hat{E}$ is a matching on the vertex set $K$ if $|K|$ is  even otherwise, $\hat{E}$ is a maximum matching on $K$ along with one extra edge. Optimality of $\hat{E}$ follows as $\hat{E}$ is a maximum matching on the vertex set $K$.

It can be observed that $\lfloor \frac{w(S)}{2} \rfloor \leq |\hat{E}| \leq \lceil \frac{w(S)}{2} \rceil$. Also, if $S^*$ is a minimum weight $T$-dominating set in $(H, w, A)$ and $\hat{E}^*$ is a minimum size solution to \textsc{Min-to-$DD_2$} for $G$ then $\lfloor \frac{w(S^*)}{2} \rfloor \leq |\hat{E}^*| \leq \lceil \frac{w(S^*)}{2} \rceil$. Hence, we have the following inequality, $\frac{|\hat{E}|}{|\hat{E}^*|} \leq \frac{w(S)}{w(S^*)} \leq O(\log |V|)$.
\end{proof}

\section{Complexity on bounded degree graphs}
In this section, we  prove that \textsc{Min-$DD_2$} and \textsc{Max-$DD_2$} are \textsc{APX}-hard for graphs with maximum degree 4.  We prove these results by establishing $L$-reductions from \textsc{Min-VC} and \textsc{Max-IS} for 3-regular graphs, respectively.
$L$-reduction is defined as follows.

\begin{definition} \cite{papaYa}
 Given two NP optimization problems $\pi_1$ and $\pi_2$ and a polynomial time transformation $f$ from instances of $\pi_1$ to instances of $\pi_2,$ we say that $f$ is an $L$-reduction if there are positive constants $\alpha$ and $\beta$ such that for every instance $x$ of $\pi_1:$
 \begin{itemize}
     \item $opt_{\pi_2}(f(x))\leq \alpha. opt_{\pi_1}(x).$
     \item for every feasible solution $y$ of $f(x)$ with objective value $m_{\pi_2}(f(x),y)$, we can find a solution $y'$ of $x$ in polynomial time with $m_{\pi_1}(x,y')$ such that $|opt_{\pi_1}(x)-m_{\pi_1}(x,y')|\leq \beta. |opt_{\pi_2}(f(x))-m_{\pi_2}(f(x),y)|.$
 \end{itemize} 
\end{definition}

\begin{theorem} \label{MinDD_2APXhard}
\textsc{Min-$DD_2$} is \textsc{APX}-hard for graphs with maximum degree $4$.
\end{theorem}
\begin{proof}
From a $3$-regular graph $G=(V, E)$, an instance of \textsc{Min-VC}, we construct $H=(V_H,E_H)$, an instance of \textsc{Min-$DD_2$}, in the following way. After making a copy of $G$, we replace each edge $uv \in E$ with a pair of edges $ue_{uv}$ and $e_{uv}v$ by introducing a new vertex $e_{uv}$. Then we add a gadget $H_{v_i}$ for every vertex $v_i\in V,$ where $H_{v_i}=(\{p_i,q_i,r_i,s_i,t_i\}, \{p_iq_i, q_ir_i, r_is_i,$ $s_ip_i, s_it_i,\})$. Finally, we introduce edge $t_iv_i$ for each vertex $v_i \in V$. Formally, $H$ is the graph in which $V_H=V\cup \{e_{uv} \mid uv\in E\}\cup \bigcup_{v_i\in V}\{p_i,q_i,r_i,s_i,t_i\}$ and $E_H=\{ue_{uv},e_{uv}v \mid uv\in E\}\cup E_V,$ where $E_V=\bigcup_{v_i\in V}\{p_iq_i,q_ir_i,r_is_i,s_ip_i,s_it_i,t_iv_i\}.$ Clearly, $H$ has maximum degree $4$ with $|V_H|=6|V| + \frac{3|V|}{2}$ and $|E_H|=2|E|+6|V|.$ Hence, $H$ can be constructed in polynomial time. It is easy to observe that $H$ is a $DD_2$-graph, as the degree of each vertex in $H$ is at least 2 (by Lemma \ref{lemma2.2}). 
For an illustration of the construction of $H$, we refer to Figure \ref{min_dd_2_apx_hard}.

\begin{figure}[htbp]
  \centering
     \includegraphics[width=11cm, height=7 cm]{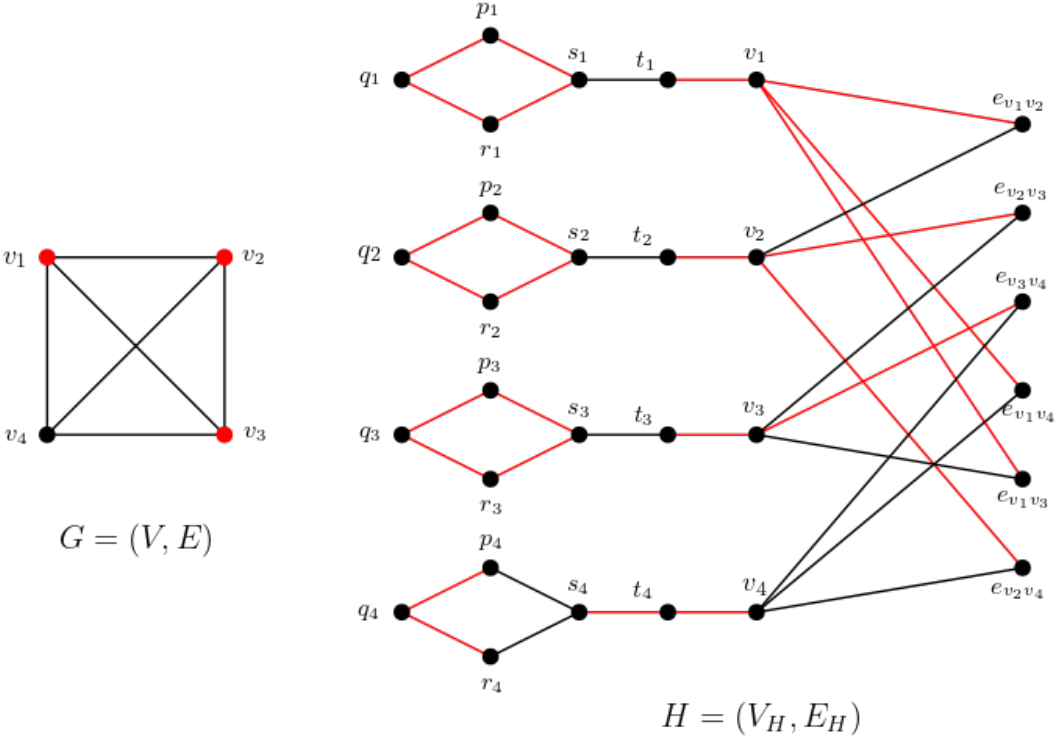}
\caption{Illustrations of the graph $H$ constructed from $G$ in the proof of Theorem \ref{MinDD_2APXhard}. 
The red colored edges in $H$ forms a minimal $DD_2$-subgraph $H_S$ corresponding to the minimal vertex cover $S$ in $G$ and the ordering $<v_1, v_2, v_3>$ of $S$.}
\label{min_dd_2_apx_hard}
 \end{figure}

First, we will show that for each minimal vertex cover $S$ of $G$ one can construct a minimal spanning $DD_2$-subgraph $H_S=(V_H, E_S)$ of $H$ such that $|E_S| = \frac{11n}{2} + |S|.$

Let $S$ be a minimal vertex cover in $G$ with $|S| = k$. Let us assume that the vertices in $S$ are ordered as <$v_{i_1}, v_{i_2}, \ldots v_{i_k}$>. From $S$, we construct a subgraph $H_S$ of $H$ whose components are defined as follows. First, we define $k$ vertex sets $JS_{v_{i_j}} = \{x \in V \mid v_{i_j}x \in E \mbox{~and~} x \notin \{v_{i_1}, \ldots, v_{i_{j-1}}\}\}$, for $j \in [k].$ It can be observed that $1 \leq |JS_{v_{i_j}}| \leq 3$ (as $S$ is a minimal vertex cover in $G$ and $G$ is 3-regular). 
 
For each $v_{i_j} \in S$, we define subgraphs of $H$ as follows. 
\begin{enumerate}
\item[•] $FS_{v_{i_j}}^1=(V_{v_{i_j}}^1,E_{v_{i_j}}^1),$ where 
 $V_{v_{i_j}}^1=\{p_{i_j},q_{i_j},r_{i_j},s_{i_j}\}, 
 E_{v_{i_j}}^1=\{p_{i_j}q_{i_j},q_{i_j}r_{i_j},r_{i_j}s_{i_j},s_{i_j}p_{i_j}\};$ 
\item[•] $FS_{v_{i_j}}^2=(V_{v_{i_j}}^2,E_{v_{i_j}}^2),$ where 
 $V_{v_{i_j}}^2=\{t_{i_j}, v_{i_j}\} \cup \{e_{v_{i_j}x} \mid x \in JS_{v_{i_j}}\}, 
 E_{v_{i_j}}^2=\{t_{i_j}v_{i_j}\} \cup \{v_{i_j}e_{v_{i_j}x} \mid x \in JS_{v_{i_j}}\};$ 
\end{enumerate}
For each $v_i \in V \setminus S$, we define subgraphs of $H$ as follows.
\begin{enumerate}
\item[•] $FS_{v_i}^3=(V_{v_i}^3,E_{v_i}^3),$ where 
 $V_{v_i}^3=\{p_i,q_i,r_i\}, E_{v_i}^3=\{p_iq_i,q_ir_i\};$
\item[•] $FS_{v_i}^4=(V_{v_i}^4,E_{v_i}^4),$ where 
 $V_{v_i}^4=\{v_i,t_i,s_i\}, E_{v_i}^4=\{v_it_i,t_is_i\}.$
\end{enumerate}


It is easy to observe that, for each ${v_{i_j}} \in S$, $FS^1_{v_{i_j}}$ is isomorphic to $C_4$ and $FS^2_{v_{i_j}}$ is isomorphic to $K_{1,z_{i_j}}$, where $z_{i_j} = |JS_{v_{i_j}}|+1$. For each $v_i \in V \setminus S$, $FS^3_{v_i}$ and $FS^4_{v_i}$ are  isomorphic to $K_{1,2}$. Hence, these subgraphs are minimal $DD_2$-subgraphs of $H$. Also, the union of these subgraphs forms a minimal spanning $DD_2$-subgraph $H_S$ of $H$ with $|E_S| = 5|S| + 4(n- |S|) + \frac{3n}{2} = \frac{11n}{2}+k$.

Next, we show that a minimum spanning $DD_2$-subgraph  $(V_H, E')$ of $H$ satisfies some properties. From a minimal spanning $DD_2$-subgraph $(V_H, E')$ of $H$ having these properties we can compute a minimal vertex cover $S$ in $G$ with $|S| = \frac{11n}{2} -|E'|$.

Let $(V_H, E')$ be a minimal spanning $DD_2$-subgraph of $H$ such that, for each $v_i \in V$, the edge set $EC_4(v_i)=\{p_iq_i, q_ir_i, r_is_i, s_ip_i\}$ is not a subset of $E'$. Then it can be observed that $E' = E_H \setminus \{s_ip_i, s_ir_i \mid v_i \in V\}$ and $|E'| = 7n$. In this case, $(V_H, E')$ has $n+1$ components with $n$ components isomorphic to $P_3$, and the largest component is a subdivision of a corona graph. 
Such minimal spanning $DD_2$-subgraphs of $H$ are not of minimum size. This is because the minimal spanning $DD_2$-subgraphs  $(V_H, E_S)$ constructed from a minimal vertex cover $S$ in $G$ satisfies $|E_S| = \frac{11n}{2} +|S| < 7n$. 

Now onward, we will consider the minimal spanning $DD_2$-subgraphs $(V_H, E')$ of $H$ such that for at least one $v_i \in V$, $EC_4(v_i) \subseteq E'$. 
Now, given such a set $E'$ we construct $S' = \{v_i \mid EC_4(v_i) \subseteq E'\}$ and we will say that $S'$ is associated with the set $E'$. If $S'$ is a vertex cover of $G$, then we are done. Otherwise, in polynomial time we will construct a new minimal spanning $DD_2$-subgraph $(V_H, E'')$ of $H$ such that $|E''| < |E'|$ and the associated set $S''$ is a vertex cover in $G$.
Because of this algorithm, we will consider minimal spanning $DD_2$-subgraphs $(V_H, E')$ of $H$ such that the associated set $S'$ is a vertex cover in $G$.

Suppose $(V_H, E')$ is a minimal spanning $DD_2$-subgraph of $H$ such that $S'$ (associated with $E'$) is not a vertex cover in $G$. Then $G[V\setminus S']$ has at least one edge and let it be $v_iv_j$. Let $T$ be the connected component in $(V_H, E')$ having the vertex $v_i$. It can be observed that $T$ is either a path $P(s_i, s_j)=(s_i, t_i, v_i, e_{v_iv_j}, v_j, t_j,s_j)$ with seven vertices (a subdivision of $P_4$ which is also a corona graph) or a subdivision of a corona graph. From this structural property of $T$ it follows that if $v_k \in N_G(v_i)\setminus \{v_j\}$ be a vertex in $T$ then the path $P(s_i, s_k) = (s_i, t_i, v_i, e_{v_iv_k}, v_k, t_k,s_k)$ is also in $T$. Now, we update the edge set $E'$ by adding two edges $p_is_i$ and $r_is_i$ to $E'$, deleting the edges $s_it_i, e_{v_iv_j}v_j$ from $E'$, and deleting the edge $e_{v_iv_k}v_k$ from $E'$, for each such vertex $v_k \in N_G(v_i)$ in $T$. Note that by this process, the size of $E'$ does not increase, and the vertex $v_i$ is included in the set $S'$. We will continue updating $E'$ and $S'$ until $S'$ becomes a vertex cover in $G$. In this process, we are breaking the component $T$ into new components, with each one a minimal $DD_2$-subgraphs. This would imply that at the end we obtain a minimal set $E'$ such that $(V_H, E')$ is a minimal spanning $DD_2$-subgraph of $H$ with $S'$ a vertex cover in $G$ and $|E'| = \frac{11n}{2} + |S'|$. Since this updating process can continue at most $\frac{3n}{4}$ updates (as a minimal vertex cover has size at most $\frac{3n}{4}$ in a 3-regular graph), it is a polynomial time algorithm.

Based on these observations, we conclude that if $E^*$ is a minimum spanning $DD_2$-subgraph of $H$ and the corresponding set $S^*$ is a minimum vertex cover in $G$ then $|E^*| = \frac{11n}{2} + |S^*| \leq 23|S^*|$ (as $|S^*| \geq \frac{n}{4}$).
Also, for any minimal spanning $DD_2$-subgraph $E'$ of $H$, we have $|S'| -|S^*| = |E'| - |E^*|$. These two inequalities show that the above reduction is an $L$-reduction with $\alpha=23$ and $\beta=1.$ Therefore, \textsc{Min-$DD_2$} is \textsc{APX}-hard for graphs with maximum degree at most $4$.
\end{proof}

From Theorem \ref{MinDD_2APXhard} and Corollary \ref{3-reg-APX}, the following corollary holds. 
\begin{corollary}
The \textsc{Min-$DD_2$} Problem is \textsc{APX}-complete for graphs with maximum degree at most $4$.
\end{corollary}

\begin{theorem}
\label{MaxDD_2APXhard}
\textsc{Max-$DD_2$} is \textsc{APX}-hard for graphs with maximum degree $4$.
\end{theorem}
\begin{proof}
We prove this theorem with the help of a reduction from \textsc{Max-IS} which is similar to the reduction given in Theorem \ref{MinDD_2APXhard}.

From a $3$-regular graph $G=(V, E)$, an instance of \textsc{Max-IS}, we construct $H=(V_H,E_H)$, an instance of \textsc{Max-$DD_2$}, in the following manner. After making a copy of $G$, we replace each edge $uv \in E$ with a pair of edges $ue_{uv}$ and $e_{uv}v$ by introducing a new vertex $e_{uv}$. Then we add a gadget $H_{v_i}$ for every vertex $v_i\in V,$ where $H_{v_i}=(\{p_i,q_i,r_i,s_i\}, \{p_iq_i, q_ir_i, r_is_i\}).$ Finally, we introduce edge $s_iv_i$ for each vertex $v_i \in V$. Formally, $H$ is the graph in which $V_H=V\cup \{e_{uv} \mid uv\in E\}\cup \bigcup_{v_i\in V}\{p_i,q_i,r_i,s_i\}$ and $E_H=\{ue_{uv},e_{uv}v \mid uv\in E\}\cup E_V,$ where $E_V=\bigcup_{v_i\in V}\{p_iq_i,q_ir_i,r_is_i,s_iv_i\}.$ Clearly, $H$ has maximum degree $4$ with $|V_H|=5|V| + \frac{3|V|}{2}$ and $|E_H|=2|E|+4|V|.$ Hence, $H$ can be constructed in polynomial time. It is easy to observe that $H$ is a $DD_2$-graph, as every weak support vertex in $H$ is adjacent to at least one vertex that is neither a pendant nor a support vertex (by Proposition \ref{charac-DD-2-weaksupport-certified}). 

\begin{figure}[htbp]
  \centering
     \includegraphics[width=10cm, height=6.5cm]{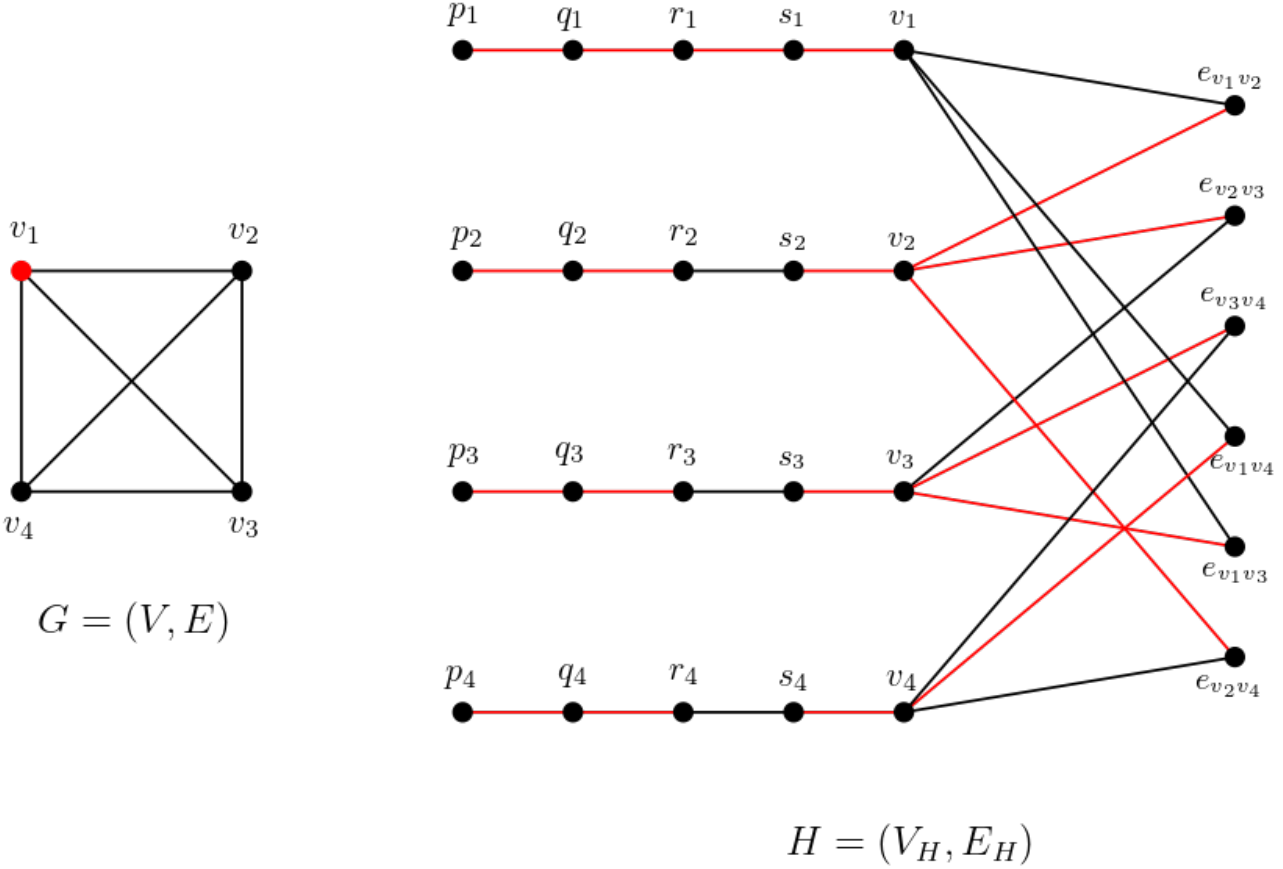}
\caption{
Illustrations of the graph $H$ constructed from $G$ in the proof of Theorem \ref{MaxDD_2APXhard}. 
The red colored edges in $H$ forms a minimal $DD_2$-subgraph $H_I$ of corresponding to the independent set $I=\{v_1\}$ in $G$ and the ordering $<v_2, v_3, v_4>$ of $V \setminus I$.}
\label{max_dd_2_apx_hard}
\end{figure}

First, we will show that for each maximal independent set $I$ of $G$, one can construct a minimal spanning $DD_2$-subgraph $H_I=(V_H, E_I)$ of $H$ such that $|E_I| = \frac{9n}{2} + |I|.$

Let $I$ be a maximal independent set of $G$ with $|I| = k$. Let us assume that the vertices not in $I$ are ordered as <$v_{i_1}, v_{i_2}, \ldots v_{i_{n-k}}$>. From $I$, we construct a subgraph $H_I$ of $H$ whose components are defined as follows. First, we define $n-k$ vertex sets $JS_{v_{i_j}} = \{x \in V \mid v_{i_j}x \in E \mbox{~and~} x \notin \{v_{i_1}, \ldots, v_{i_{j-1}}\}\}$, for $j \in [n-k]$. It can be observed that $1 \leq |JS_{v_{i_j}}| \leq 3$ (as $I$ is a maximal independent set in $G$ and $G$ is 3-regular). 
 
For each $v_{i_j} \in V\setminus I$, we define subgraphs of $H$ as follows. 
\begin{enumerate}
\item[•] $FS_{v_{i_j}}^1=(V_{v_{i_j}}^1,E_{v_{i_j}}^1),$ where 
 $V_{v_{i_j}}^1=\{p_{i_j},q_{i_j},r_{i_j}\}, 
 E_{v_{i_j}}^1=\{p_{i_j}q_{i_j},q_{i_j}r_{i_j}\};$ 
\item[•] $FS_{v_{i_j}}^2=(V_{v_{i_j}}^2,E_{v_{i_j}}^2),$ where 
 $V_{v_{i_j}}^2=\{s_{i_j}, v_{i_j}\} \cup \{e_{v_{i_j}x} \mid x \in JS_{v_{i_j}}\}, 
 E_{v_{i_j}}^2=\{s_{i_j}v_{i_j}\} \cup \{v_{i_j}e_{v_{i_j}x} \mid x \in JS_{v_{i_j}}\};$ 
\end{enumerate}
For each $v_i \in I$, we define subgraphs of $H$ as follows.
\begin{enumerate}
\item[•] $FS_{v_i}^3=(V_{v_i}^3,E_{v_i}^3),$ where 
 $V_{v_i}^3=\{p_i,q_i,r_i,s_i,v_i\}, E_{v_i}^3=\{p_iq_i,q_ir_i,r_is_i,s_iv_i\};$
\end{enumerate}

It is easy to observe that, for each $v_{i_j} \in V\setminus I$, $FS^1_{v_{i_j}}$ is isomorphic to $K_{1,2}$ and $FS^2_{v_{i_j}}$ is isomorphic to $K_{1,z_{i_j}}$, where $z_{i_j} = |JS_{v_{i_j}}|+1$. For each $v_i \in I$, $FS^3_{v_i}$ is isomorphic to subdivision of a corona graph $K_{1,2}$. Hence, these subgraphs are minimal $DD_2$-subgraphs of $H$. Also, the union of these subgraphs forms a minimal spanning $DD_2$-subgraph $H_I$ of $H$ with $|E_I| = 4|I| + 3(n- |I|) + \frac{3n}{2} = \frac{9n}{2}+|I|$.

The above observation shows that $H$ is not a minimal spanning $DD_2$-graph. Let $(V_H, E')$ be a minimal spanning $DD_2$-subgraph of $H$. This subgraph satisfies the following properties. (a) For each vertex $v_i \in V$, the vertex $p_i$ must be in a component which is either a path $P(p_i, r_i) = (p_i, q_i, r_i)$ or a path $P(p_i, v_i)=(p_i, q_i, r_i, s_i, v_i)$ 
(because $P_3$ and $P_5$ are the only paths which are $DD_2$ graphs). (b) Each vertex $v_i$ must be in a component, either a path $P_5$ or a $K_{1, t}$ ($1 \leq t \leq 4$) with $v_i$ as a center vertex. Suppose $v_i$ is in the 
component $P(s_i, s_j)=(s_i, v_i, e_{v_i v_j}, v_j, s_j)$ (for some edge $v_iv_j \in E$), which is a $P_5$. Note that, $P(p_i, r_i)$ and $P(p_j, r_j)$
are also components in $(V_H, E')$. In this case, we update the edge set $E'$ as $(E' \cup \{r_is_i\}) \setminus \{v_ie_{v_i v_j}\}$. Here, the size of $E'$ remains unchanged, and both $v_i$ and $v_j$ do not appear in a single component in $(V_H, E')$, for all $v_iv_j \in E'$.

Given a minimal spanning $DD_2$-subgraph $(V_H, E')$ of $H$, we claim that the set $S' = \{v_i \mid P(p_i, v_i)$ is a component in $(V_H, E')\}$ is an independent set of $G$. For any edge $v_iv_j \in E$, the paths $P(s_i, v_i)$ and $P(s_j, v_j)$ can not be components in $(V_H, E')$. If so, then $(V_H, E')$ has a component with only one vertex $e_{v_i v_j}$, which is not a $DD_2$ graph. Now, it is easy to observe that $|E'| = \frac{9n}{2} + |S'|$.

Based on these observations, we conclude that if $(V_H, E^*)$ is a maximum size minimal spanning $DD_2$-subgraph of $H$ then the corresponding set $I^*$ is a maximum independent set in $G$ and $|E^*| = \frac{9n}{2} + |I^*| \leq 19|I^*|$ (as $|I^*| \geq \frac{n}{4}$).
Also, for any minimal spanning $DD_2$-subgraph $(V_H, E')$ of $H$, we have $|I'| -|I^*| = |E'| - |E^*|$. These two inequalities show that this is an $L$-reduction with $\alpha=19$ and $\beta=1.$ Therefore, \textsc{Max-$DD_2$} is \textsc{APX}-hard for graphs with maximum degree at most $4$.
\end{proof}

From Theorem \ref{MaxDD_2APXhard} and Corollary \ref{3-reg-APX}, the following corollary holds: 
\begin{corollary}
The \textsc{Max-$DD_2$} Problem is \textsc{APX}-complete for graphs with maximum degree at most $4$.
\end{corollary}

\section{Lower bound on approximability of \textsc{Max-Min-to-$DD_2$}}
Here, we will prove a strong inapproximability result for \textsc{Max-Min-to-$DD_2$}. We obtain this result by a reduction from \textsc{Max-Min-VC}. Given a graph $G=(V, E)$, in \textsc{Max-Min-VC} it is required to find a vertex set $S$ in $G$ of maximum cardinality such that $S$ is a minimal vertex cover in $G$. We will use the known lower bound result on \textsc{Max-Min-VC}.

%

\begin{theorem} \cite{boria2015max} \label{max_min_VC}
For any $\varepsilon > 0$, \textsc{Max-Min-VC} can not be approximated within a ratio of  $n^{\frac{1}{2}-\varepsilon}$, unless \textsc{P=NP}.  
\end{theorem}

\begin{theorem}
For any $\varepsilon >0,$ \textsc{Max-Min-to-$DD_2$} can not be approximated within a ratio of $n^{\frac{1}{6}-\varepsilon}$ for bipartite graphs, unless \textsc{P=NP}. 
\end{theorem}
\begin{proof}
Given a graph $G=(V,E)$ an instance of \textsc{Max-Min-VC}, the construction of $G'=(V',E')$ an instance of \textsc{Max-Min-to-$DD_2$}, is as follows. Assume that $V=\{v_1,v_2,\ldots,v_n\}$ and $E=\{e_1,e_2,\ldots,e_m\}.$ 
\begin{itemize}
    \item Make a copy of $V$. For each vertex $v_i \in V$, create $m+1$ new vertices $a_i^1, a_i^2, \ldots, a_i^{m+1}$ and $(m+1)$ edges $v_{i}a_i^1, \ldots, v_{i}a_i^{m+1}$.

    \item For each edge $e=uv\in E,$ we create two vertices $e_{uv}$, $l_{uv}$ and introduce three edges $e_{uv}l_{uv}, ue_{uv}, e_{uv}v$.
    
    \item Finally, we add a new vertex $p$ and make it adjacent to each vertex $v_i \in V$.
\end{itemize}
This completes the construction of $G'$. For an illustration of $G'$, we refer to Figure \ref{max-min-dd-2-add}. It is easy to observe that $G'$ can be constructed in polynomial time as $|V'|=O(|V|^3)$ and $|E'|=O(|V|^3)$. Observe that the constructed graph $G'$ is a non-$DD_2$-bipartite graph as each weak support vertex $e_{uv}$ does not satisfy the $2^{nd}$ property in Proposition \ref{charac-DD-2-weaksupport-certified} (i.e., $N_{G'}(e_{uv}) \setminus (L_{G'} \cup S_{G'}) \neq \emptyset$).

\begin{figure}[htbp]
    \centering
    \includegraphics[width=13 cm, height=5 cm]{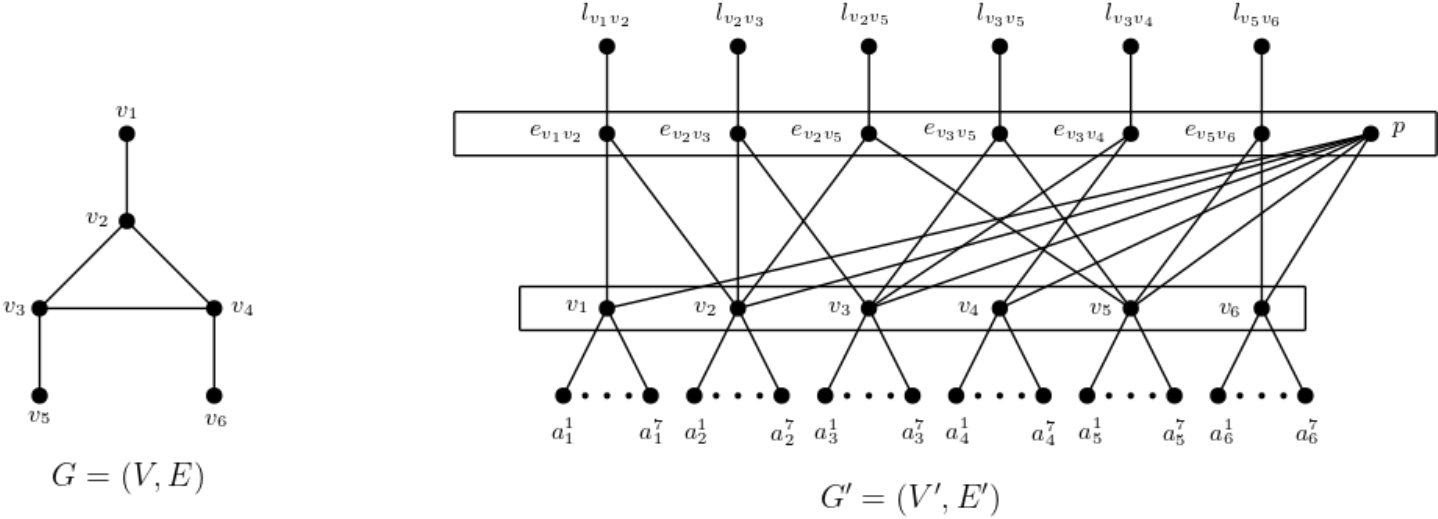}    \caption{An illustration of the above reduction through an example}
    \label{max-min-dd-2-add}
\end{figure}

\begin{claim}
\label{Claim-add}
The graph $G$ has a minimal vertex cover of cardinality at least $k$ if and only if there exists a solution to the \textsc{Max-Min-to-$DD_2$} problem for $G'$ of size at least $k(m+1)$.
\end{claim}
\begin{proof}
Let $S$ be a minimal vertex cover of $G$ with cardinality at least $k.$ It can be easily verified that the edge set $\hat{E}= \bigcup\limits_{v_i\in S} \{a_i^jp \mid \forall j\in [m+1]\}$ is a solution to \textsc{Max-Min-to-$DD_2$} of size at least $k(m+1)$ for the graph $G'$. This is because in graph $(V', E' \cup \hat{E})$, every weak support vertex has at least one neighbor, which is neither a pendant vertex nor a support vertex.

Conversely, suppose $\hat{E}$ be a minimal set of edges added to the graph $G'$ such that $H=(V, E' \cup \hat{E})$ is a $DD_2$-graph of cardinality  at least $ k(m+1)$. From Proposition \ref{charac-DD-2-weaksupport-certified}, there exists some weak support vertices (say $w$) in $G'$ such that $N_{G'}(w)\setminus (L_{G'}\cup S_{G'}) = \emptyset$. Thus, in graph $H$, to make $N_H(w)\setminus (L_H\cup S_H) \neq \emptyset$ for such $w\in S_{G'}$, one of the following must hold: (a) the vertex $w$ is adjacent to some vertex (say $x\notin S_{G'}$) (this can be done by adding an edge $wx$); (b) the vertex $w$ is no more a weak support vertex (this can be done by adding an edge to the unique pendent vertex adjacent to $w$); or (c) at least one of the strong support vertex (say $u$) adjacent to $w$ is no more a support vertex (this can be done by adding edges to all the pendent vertices adjacent to $u$).

Now let $w\in S'_{G'}$ be a vertex such that $N_{G'}(w)\setminus (L_{G'}\cup S_{G'}) = \emptyset$ and by using (a) $N_H(w)\setminus (L_H\cup S_H) \neq \emptyset$ in the graph $H.$ Then let $E^* = (\hat{E}\setminus \{wx\})\cup \{a_t^{i}p, \forall i\in [m+1]\}$, where $x(\notin S_{G'})$ is a vertex of $G'$ and $\{a_t^{i}, \forall i\in [m+1]\}$ be the set of pendents adjacent to $t$ ($t$ is a strong support vertex adjacent to $w$), is also a minimal set of edges of larger cardinality than $\hat{E}$ such that $H=(V, E' \cup E^*)$ is a $DD_2$-graph, which is a contradiction to our assumption. If using (b) $N_H(w)\setminus (L_H\cup S_H) \neq \emptyset$ in the graph $H$ then $E^* = (\hat{E}\setminus \{l_wx\})\cup \{a_{t_w}^{i}p, \forall i\in [m+1]\}$ if $x \notin L'_{G'}$ and $E^* = (\hat{E}\setminus \{l_wx\})\cup \{a_{t_w}^{i}p, a_{t_y}^{i}p, \forall i\in [m+1]\}$ if $x \in L'_{G'},$ where $\{a_{t_w}^{i},a_{t_y}^{i}p, \forall i\in [m+1]\}$ be the set of pendents adjacent to $t_w$ and $t_y$ respectively, ($t_w$ and $t_y$ are some strong support vertex adjacent to $w$ and $y$ ($y \in S'_{G'}$ is the vertex adjacent to $x\in L'_{G'}$) respectively), is also a minimal set of edges of larger cardinality than $\hat{E}$ such that $H=(V, E' \cup E^*)$ is a $DD_2$-graph, which is also a contradiction.

Hence assume $N_H(w)\setminus (L_H\cup S_H) \neq \emptyset$ holds for all $w\in S'_{G'}$ by using (c). Now let the required vertex set is $S=\{v_i\in V \mid a_i^jp \in \hat{E} ~\mbox{for all}~ j\in [m+1])\}.$ Now claim that $S$ is a vertex cover of $G.$ If not, then there exists an edge $e(=uv)\in E$ such that neither $u\in S$ nor $v\in S,$ which implies $N_H(e_{uv})\setminus (L_H\cup S_H) \neq \emptyset$ for $e_{uv}\in S'_{G'}$, by using either (a) or (b), then the resulting set of added edges $E^*$ is of larger cardinality than $|\hat{E}|=k(m+1),$ a contradiction. Therefore, $S$ is a vertex cover of $G$ and $|S|\geq k$ as $|\hat{E}|\geq k(m+1).$ 
\end{proof}

From the above observations, it follows that $|\hat{E}^*| = (m+1)|S^*|$, where $S^*$ is a maximum size minimal vertex cover in $G$ and $\hat{E}^*$ is a maximum size minimal edge addition set to the graph $G'$. 

Let us assume that there exists some fixed $\varepsilon' >0$ such that \textsc{Max-Min-to-$DD_2$} problem for graphs with $n'$ vertices can be approximated within a ratio of $\alpha=n'^{\frac{1}{6}-\varepsilon'}$ by using an algorithm  
${A}_{DD_2}$ that runs in polynomial time.

Then $|\hat{E}'|\leq \alpha |\hat{E}^*|,$ where $\alpha=n'^{\frac{1}{6}-\varepsilon'}$ and note that $|\hat{E}^*|=|S^*|(m+1)$.

Now, $|S|\leq \frac{|\hat{E}|}{m+1} \leq \alpha \frac{|\hat{E}^*|}{(m+1)}= \alpha |S^*|,$ where $\alpha=n'^{\frac{1}{6}-\varepsilon'}\leq (c n^3)^{(\frac{1}{6}-\varepsilon')}\leq c^{(\frac{1}{6}-\varepsilon')} n^{(\frac{1}{2}-3\varepsilon')},$ for some constant $c>0$.

Choose $\varepsilon>0,$ such that $c^{\frac{1}{6}-\varepsilon'}<n^{3\varepsilon'-\varepsilon},$ then $|S|\leq \alpha |S^*|< n^{3\varepsilon'-\varepsilon}n^{\frac{1}{2}-3\varepsilon'}|S^*|=n^{\frac{1}{2}-\varepsilon}|S^*|.$ 

Hence, $|S|\leq n^{\frac{1}{2}-\varepsilon}|S^*|,$ which leads to a contradiction to the Theorem \ref{max_min_VC}. Therefore, the \textsc{Max-Min-to-$DD_2$ problem} for a bipartite graph $G=(V,E)$ can not be approximated within $n^{\frac{1}{6}-\varepsilon}$ for any $\varepsilon >0,$ unless \textsc{P=NP}.     
\end{proof}

\section{Conclusion}
It would be interesting to design approximation algorithms for \textsc{Min-DD$_2$} and \textsc{Max-DD$_2$} with approximation factor smaller than 3. We suspect that these two problems are $APX$-complete for 3-regular graphs.

\end{document}